\documentclass[11pt]{article}
\usepackage{graphicx,hyperref}
\usepackage{amssymb,amsmath,amsthm}
\usepackage{fullpage}
\usepackage[left]{lineno}
\usepackage[]{hyperref}
\newcommand{\comment}[1]{} 

\newtheorem{theorem}{Theorem}
\newtheorem{remark}[theorem]{Remark}
\newtheorem{lemma}[theorem]{Lemma}
\newtheorem{corollary}[theorem]{Corollary}

\theoremstyle{definition}
\newtheorem{definition}[theorem]{Definition}

\usepackage{microtype} 
\usepackage{color}

\newtheorem{problem}[theorem]{Problem}

\newcommand\blfootnote[1]{%
  \begingroup
  \renewcommand\thefootnote{}\footnote{#1}%
  \addtocounter{footnote}{-1}%
  \endgroup
}

\title{Scaling and compressing melodies using geometric similarity measures}

\author{L.E. Caraballo
\thanks{Department of Applied Mathematics II, University of Seville, Camino de los descubrimientos s/n, 41092 Seville, Spain. Email: varo.caraballo87@gmail.com, \{dbanez, frodriguex, vscanales, iventura\}@us.es}
\and
J.M. D\'iaz-B\'a\~nez $^*$
\and F. Rodr\'iguez $^*$
\and V. S\'anchez-Canales $^*$
\and I. Ventura $*$\thanks{Corresponding author.}
}

\begin{document}

\maketitle

\begin{abstract}

Melodic similarity measurement is of key importance in music information retrieval. In this paper, we use geometric matching techniques to measure the similarity between two melodies. We represent music as sets of points or sets of horizontal
line segments in the Euclidean plane and propose efficient algorithms for optimization problems inspired in two operations on melodies; linear scaling and audio compression. In the scaling problem, an incoming query melody is scaled forward until the similarity measure between the query and a reference melody is minimized. The compression problem asks for a subset of notes of a given melody such that the matching cost between the selected notes and the reference melody is minimized.

	\vspace{.5cm}
{\bf Keywords:} Melodic similarity,   geometric   matching, algorithm, scaling, compressing.

\blfootnote{\begin{minipage}[l]{0.3\columnwidth} \vspace{-3mm}\hspace{-0.41cm} \includegraphics[trim=10cm 6cm 10cm 5cm,clip,scale=0.15]{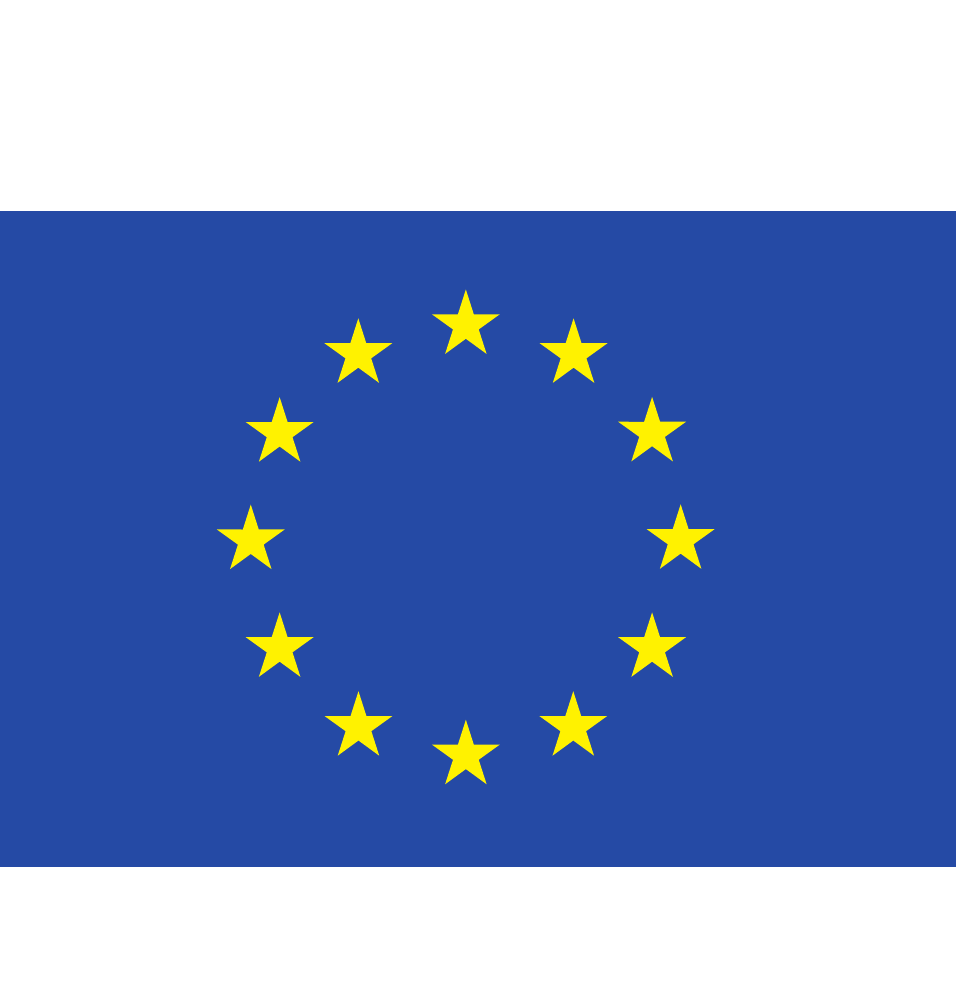} \end{minipage}
\hspace{-3.cm} \begin{minipage}[l][1cm]{0.88\columnwidth} 
		\vspace{-1mm}
		This research has been
 supported by projects MTM2016-76272-R of the Spanish Ministry of Economy and Competitiveness  (AEI/FEDER,UE), PID2020-114154RB-I00 of the  Ministry of Science and Innovation (CIN/AEI/10.13039/501100011033/)
  and the European Union's Horizon 2020 research and innovation programme under the Marie Sk\l{}odowska-Curie grant agreement No.\ 734922.
Research of V. S\'anchez-Canales was partially supported by
 San Pablo  Andalucia CEU Foundation.
\end{minipage}}

   \end{abstract}

\section{Introduction}

Musical information retrieval (MIR)
is a flourishing
field that takes advantage of
techniques developed in computer science and mathematics.
With the increased
availability of large musical data comes the
need to design methods for searching and analysing this
information. Musicological and computational studies on rhythmic and melodic similarity have given rise to a number of geometric problems \cite{toussaint2010computational,barba2016asymmetric,bereg2019computing}.
A melody can be codified as a consecutive sequence of musical notes and each note can be represented by a point (point representation) \cite{clifford2006} or a horizontal segment (segment representation) with a time/pitch value \cite{maidin1998}.
In this paper we study two problems that arise in MIR;  linear scaling and melody compressing.
Linear scaling is used for tempo variation and melody matching \cite{cao2009phrase,jang2001content}, while an important application of data compression is clustering and classification
\cite{cilibrasi2004}.
Given a reference melody, a query melody, and a similarity measure, in the \emph{scaling problem} the incoming query is scaled forward in the horizontal direction
to find the minimum similarity measure between it and the reference melody. The \emph{compression problem} seeks to select $k$ notes of a given melody such that the similarity measure between the reference melody and the simplified melody is minimized.

The main techniques that have been explored in the literature for melodic similarity are based on string matching and geometric measures.
String matching techniques are more efficient but techniques based on geometric similarity measures have a higher
accuracy search rate.

Some strategies based on string matching are longest common subsequence (LCS)-based algorithms, dynamic time warping (DTW) and edit distance (ED)-based methods. The LCS approach represents a global alignment since it does not characterize local information such as gaps or mismatched symbols of the common sequence. To overcome this problem, most LCS-based algorithms use penalty functions \cite{suyoto2008searching}. DTW adapts two sequences non-linearly in the time dimension to calculate the similarity, and it has been extensively explored for audio comparison \cite{byrd2002problems}. The edit distance can be expressed as the minimum number of edit operations to convert one sequence into the other. Edit distance--based methods have been proposed as a music dissimilarity measure \cite{mongeau1990comparison, robine2007music}, and they have been combined with the use of a specific set of representative mid-level melodic features in specific musical repertoires \cite{mora2016}.

The scaling problem in this paper assumes the symbolic melodic representation and measures similarities with a geometric matching.
A closely related problem is to compute the best matching between a short melodic query phrase and a longer reference melody. A geometric technique that considers both pitch and rhythm information was proposed by \cite{maidin1998}:
Each note is represented as a horizontal line segment, so a sequence of notes
can be described as a rectangular contour in a 2D coordinate
system, in which the horizontal and vertical coordinates
correspond to time and pitch value, respectively.
Then the similarity measure between two melodies is the minimum area between them.
The optimal matching is obtained by moving the query segment from left to right and from down to up until the matched area is minimized.
Using a binary search tree, in \cite{aloupis2006} this geometric similarity measure is computed in $O(nm(m+\log n))$ time, where $n$ and $m$ are the lengths of the reference and  the query melodies, respectively.
Later, the approach was improved by using pitch intervals and a branch-and-prune technique to discard the incorrect horizontal positions of the query \cite{lin2008geometric}. Finally, the complexity was improved to $O(nm\log n)$ by
\cite{lin2008efficient} by using a balanced binary search tree.

Audio data compression is a fundamental component of multimedia big data and a challenge addressed in the field of MIR since its inception. In fact, a major innovation that enabled  the explosive  growth  of digitally represented music is audio compression \cite{painter2000perceptual}, which is frequently used to reduce the high data requirements of audio and music.
Early works focused mainly on the MPEG audio compression standard and the extraction of timbral texture features \cite{tzanetakis2000sound}.   More recently, the use of newer sparse representations as  the  basis for  audio  feature  extraction has been  explored, covering timbral texture, rhythmic content and pitch content \cite{ravelli2010audio}. Motivated by this task,  in this paper we have introduced a new theoretical problem that, given a melody with $n$ notes and a positive integer $k<n$, asks for $k$ notes to be selected such that the cost given by a geometric matching is minimized. It is clear that the goal here is to efficiently solve this problem from a geometric optimization point of view and important aspects as psychoacoustics (i.e., the study of sound perception by humans) are not used to guide the process.

\subsection{Problem Statement}

A melody is made up of musical notes, each of which has attributes such as pitch, duration, timbre, loudness, etc., and the melody can be coded in a symbolic representation or an audio wave format.

In this paper, we adopt the symbolic representation, so a melody is represented by a sequence
of notes in a time-pitch plane. Each note can be described by a point in 2D representing the pitch at a specific time (point representation) or a horizontal
line segment, of which the height and width denote
its pitch value and duration, respectively (segment representation). In segment representation, a melody is a consecutive sequence of horizontal segments; that is,
a melody can be seen as a step function of time.
We also call this representation a melodic contour.

In either case, for the sake of key-invariance,  we will use a pitch interval sequence to represent a melody.
Note that with the key-invariance property of this representation, the query melody only needs to move horizontally for the scaling problem.
We assume that
the melodies
are given in
lexicographic order of the starting points of the line segments.

Let $R = (R_1, R_2, \dots , R_n)$ and $Q = (Q_1, Q_2, \dots , Q_m)$ be sequences representing two melodies.
$R$ is the  reference melody from the dataset and $Q$ is the query melody to be matched.
From here on, we suppose that the duration of melody $Q$ is less than that of $R$.
The elements of $R$ and $Q$ can be consecutive horizontal segments or points in 2D depending on the representation to be used.
In the point representation, we assume that the points representing the notes are just the middle points of the horizontal segments in the corresponding segment representation. However, the same approach works if the points correspond to endpoints of the segments.

In the following we introduce two geometric measures to compute similarity and two operations on a melody that are the main ingredients of the problems.

\medskip\noindent\textbf{Area measure:} The region between two melodies with the same duration in the segment representation can be partitioned into rectangles with vertical edges supported by vertical straight lines passing through ending points of the segments. The area between two melodies is defined as the sum of the areas of the rectangular regions of the partition. Since the duration of $Q$ is less than  the duration of $R$, we extend  the last segment of $Q$ so that the durations of  $R$ and $Q$ are the same. See Figure \ref{fig:scaling}a).

\medskip\noindent\textbf{t-Monotone matching  measure:}
Let $R=\{R_i=(x_i,p_i),\, i=1,\dots,n\}$ and $Q=\{Q_j=(t_j,q_j),\, j=1,\dots,m$\} be two melodies in point representation.
In a $t$-monotone matching between $R$ and $Q$, we map each point of  the reference melody $R$ to its $t$-monotone nearest point in $Q$; that is, the left or the right point in time. The unmatched points of $Q$ are associated to their $t$-monotone nearest point in $R$.
More formally, given $R_i=(x_i,p_i)\in R$: if $x_i<t_1$ then $R_i$ is assigned to $Q_1$; if $x_i>t_m$ then $R_i$ is assigned to $Q_m$; if $t_j<x_i<(t_j+t_{j+1})/2$ then $R_i$ is assigned to $Q_j$; if $(t_j+t_{j+1})/2<x_i<t_{j+1}$ then $R_i$ is assigned to $Q_{j+1}$. Finally, if $x_i = (t_j+t_{j+1})/2$ then $R_i$ is assigned to the $l_1$-nearest point among $Q_j$ and $Q_{j+1}$ (with preference for $Q_j$ in the case of a tie), where $l_1$ denotes the Manhattan metric. For unmatched points of $Q$, we follow the analogous rule.
See Figure \ref{fig:scaling}(b). The cost of the note $Q_j$, $\phi(Q_j)$, is given by the sum of the $l_1$-distances between $Q_j$ and its matched points in $R$, and the total cost of the matching is
\[{\phi}(R,Q)=\sum_{Q_j\in Q}\phi(Q_j).\]

\noindent\textbf{\boldmath{$\varepsilon$}-scaling operation:} Consider  a segment representation of two melodies $R$ and $Q$. Let $X = (x_0=0,x_1, x_2, \dots, x_n)$ and $T = (t_0=0,t_1, t_2, \dots, t_m)$ be the time partitions given by the segment representation of $R$ and $Q$, respectively,  with $t_m \leq x_n$.
Given $\varepsilon>0$, we define the $\varepsilon$-\emph{scaling} on the query $Q$, $S_Q({\varepsilon})$, as the operation of increasing the length of each segment of $Q$ by $\varepsilon$ but keeping the starting
point of the first segment of $Q$ static. Thus, after an $\varepsilon$-scaling, $X$ does not change and $T$ is transformed to $T+\varepsilon=(t_0,t_1+\varepsilon, t_2+2\varepsilon, \dots, t_m+m\varepsilon)$.
The query can be scaled until the two melodies have the
same time duration. Thus,
$0\leq \varepsilon\leq \frac{x_n-t_m}{m}$.
Note that the pitches of the notes are unchanged in the scaling operation.
Now, for a point representation of $R$ and $Q$, $R=\{(x_i,p_i), i=1,\dots,n\}$ and $Q=\{(t_j,q_j), j=1,\dots,m$ \}, the $\varepsilon$-\emph{scaling} of $Q$ is given by
\[S_Q({\varepsilon})=\left\{\left(t_1+\frac{\varepsilon}{2},q_1\right), \left(t_2+\frac{3}{2}\varepsilon,q_2\right),\dots, \left(t_m+(m-1)\varepsilon+\frac{\varepsilon}{2},q_m\right)\right\}.\]

Let $t_j(\varepsilon)$ denote the time associated with the $j$-th note of $S_Q(\varepsilon)$. Notice that $t_j(\varepsilon)=t_j+j\varepsilon$ in the segment representation and  $t_j(\varepsilon) = t_j+\frac{2j-1}{2}\varepsilon$ in the point representation.

\medskip\noindent\textbf{\boldmath{$k$}-compressed melody operation:} Given a melody $R$ with $n$ notes in segment representation and a positive integer $k<n$, a $k$-compressed melody or $k$-compression of $R$, $C_k(R)$,
is a melody composed of $k$ segments such that $C_k(R)$
and $R$ have the same duration and each segment of $C_k(R)$
contains at least a segment of $R$. See Figure \ref{fig:compressing}(b). For the point representation, a $k$-compressed melody $C_k(R)$
is a subset of $k$ notes of $R$. See Figure \ref{fig:compressing}(a).

\begin{figure}[h!]
\centering
\includegraphics[scale=0.7]{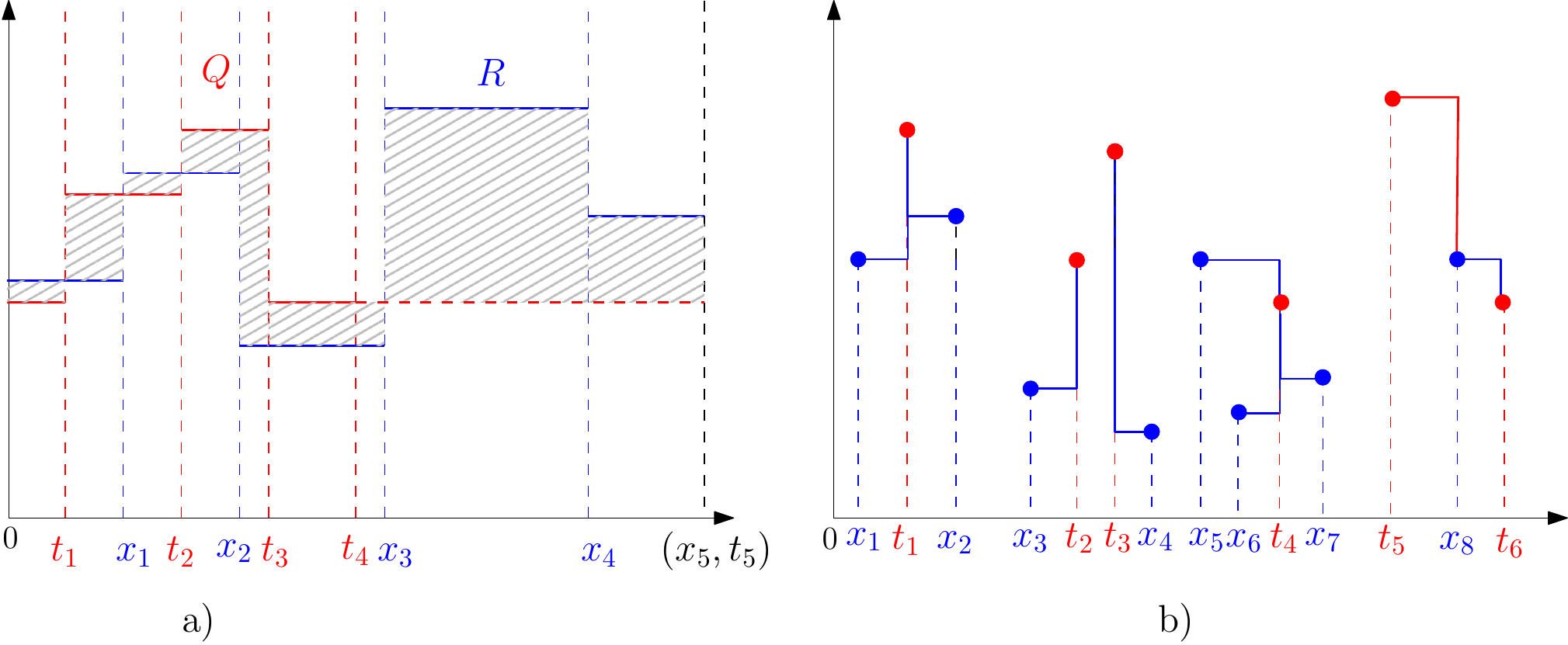}
\caption{a) Area between the reference melody $R$  and the query $Q$, Problem \ref{prob:scaling-area}.  b) $t$-Monotone matching between the reference $R$ (blue) and the query $Q$ (red), Problem \ref{prob:scaling-matching}.}
\label{fig:scaling}
\end{figure}

\begin{figure}[h!]
\centering
\includegraphics[scale=0.7]{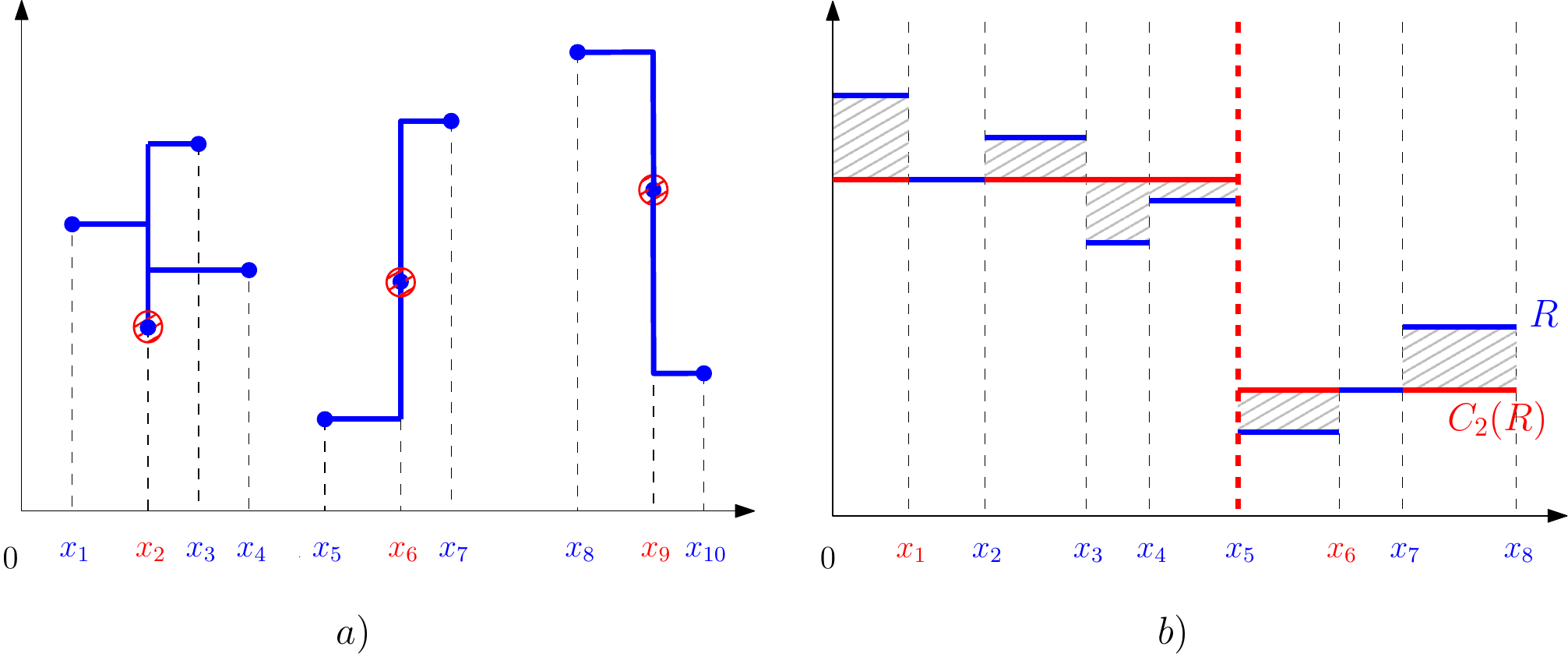}
\caption{a) A $3$-compressed melody from an input of $10$ notes, Problem \ref{prob:compressing-matching}. b) Area between a melody $R$ and a $2$-compressed melody $C_2(R)$,
Problem \ref{prob:compressing-area}. }
\label{fig:compressing}
\end{figure}

	In this paper we study the following optimization problems:
	
	\begin{problem}\label{prob:scaling-area}
	 Given two melodies $R$ and $Q$ in segment representation, compute the value of $\varepsilon>0$ such that the area between $R$ and $S_Q({\varepsilon})$ is minimized.

	\end{problem}

	\begin{problem}\label{prob:scaling-matching}
		 Given two melodies $R$ and $Q$ in point representation, compute the value of $\varepsilon>0$ such that the cost of the $t$-monotone matching between $R$ and $S_Q({\varepsilon})$ is minimized.	\end{problem}

	\begin{problem}\label{prob:compressing-matching}
Given a melody $R$ with $n$ notes in point representation  and  $k<n$, compute a $k$-compressed melody
$C_k(R)$ such that the cost of the $t$-monotone matching between $R$ and
$C_k(R)$ is minimized.
	\end{problem}
	
	\begin{problem}\label{prob:compressing-area}
Given a melody $R$ with $n$ notes in segment representation  and $k$,  $k<n$, compute a $k$-compressed melody
$C_k(R)$ such that the area between $R$ and
$C_k(R)$ is minimized.
	\end{problem}

Note that in the scaling problems,
the reference melody is static whereas
the query melody is dynamic.
However, the compressing problems ask for an optimal selection of the notes in the input melody. Figures \ref{fig:scaling} and \ref{fig:compressing} illustrate the problems.

\subsection{Results and Overview}

In this article, we tackle two geometric problems related to tasks addressed in music information retrieval. Problems 1 and 2 are solved with similar sweeping algorithms in Section \ref{sec:scaling}. Problem 2 requires more insight and is, in fact, a constrained case of what is called many-to-many problem studied in areas such as computational biology \cite{ben2003restriction} and music theory and computational geometry \cite{toussaint2004geometry}. Given two sets of two-dimensional points $R$ and $B$ with total cardinality $N$,
a \emph{many-to-many} matching between $R$ and $B$ pairs each element of $R$ to at least one element of $B$ and vice versa.
Given a cost function $c(r, b)$, with $r \in R$ and $b \in B$, the cost of a matching is defined as the sum of the costs of all matched pairs. The problem is to compute a minimum-cost many-to-many matching.
In this general setting, the problem can be solved in $O(N^3)$ time by reducing it to the minimum-weight perfect matching problem in a bipartite graph~\cite{eiter1997distance}, and no best complexity is known. We study the particular case of a $t$-monotone matching; i.e., the assignment is monotone with respect to the time direction, following the principle that the matching is made between notes that are close in time. We prove that both problems can be solved in $O(nm\log m)$.

In Section \ref{sec:comp}, we address the compression problem in both variants, using  segment and point representation. First, we apply dynamic programming to solve Problem 3, where the monotone nature of the matching is exploited. We show how to solve the problem in $O(kn^2)$ time. It is important to note that Problem 3 is a constrained case of a well known problem in operational research, the $k$-discrete median problem, defined as follows:  Given $n$ points in the plane and an integer $k$, choose $k$ out of  the $n$ input  points as medians such that the sum of the distances from the input points  to  their  nearest medians  is  minimized. This problem is shown to be NP-hard \cite{megiddo1984complexity}.

Using some remarkable properties of a solution of Problem 4, we also use dynamic programming to efficiently solve the compression problem for the segment representation in $O(k\rho n)$ time, where $\rho$ is the number of different pitch values among the notes of $R$. Note that if $\rho<<n$, the algorithm is subquadratic. Finally, in Section \ref{sec:conclu} we draw our conclusions and present some lines for future work.

\section{Scaling Operation}\label{sec:scaling}

\subsection{Area as Similarity Measure}

The area between the melodies $R$ and $Q$ is the sum of $O(m+n)$ rectangles, as illustrated in Figure \ref{fig:scaling}a).
Denote by $A_{RQ}(\varepsilon)$ the area between $R$ and $S_Q(\varepsilon)$ as a function of $\varepsilon$.
Observe that if we scale $Q$ by continuously increasing $\varepsilon$, eventually, for some value of $\varepsilon$, there exist $i$ and $j$ such that $x_i=t_j(\varepsilon)$. At this instant, some rectangles disappear and after that, new rectangles appear. We call this value of $\varepsilon$ an \emph{event}. There are $O(nm)$ events. Also note that between two events, the areas of some rectangles increase, others decrease and others are unaffected.
The type of a rectangle can be
characterised by its vertical edges. Rectangles can be classified into four types: type $C_0$, with vertical edges passing through $x_i$ and $x_{i+1}$; type
 $C_1$, with vertical edges passing through $t_j$ and $t_{j+1}$; type
 $C_2$, with vertical edges passing through $x_i$ and $t_j$ and type
 $C_3$, with vertical edges passing through $t_j$ and $x_i$. Figure \ref{fig:area_events} illustrates an event in which a rectangle of type $C_3$ disappears.

\begin{figure}[h!]
\centering
\includegraphics[scale=0.8]{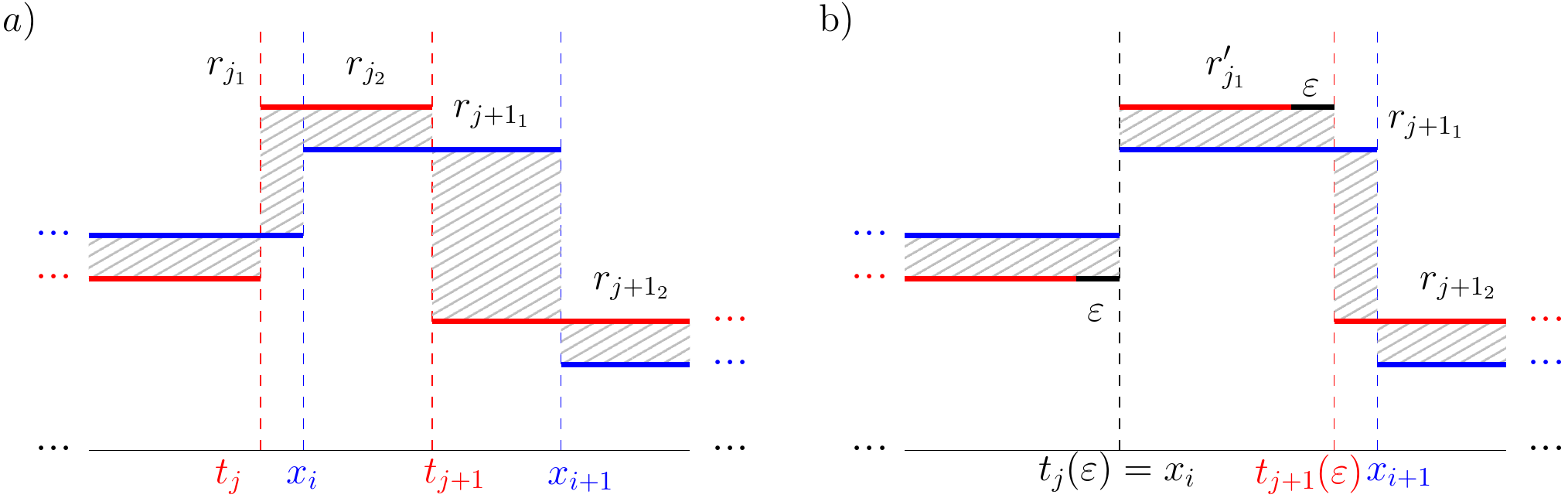}
\caption{a) Rectangles between $R$ and $Q$. b) After the scaling $S_Q(\varepsilon)$, an event is found when $t_{j}$ touches $x_{i}$. Rectangle $r_{j_1}$ disappears and $r_{j_2}$ is renamed to $r'_{j_1}$.
\label{fig:area_events}}
\end{figure}

The following result allows us to discretize the problem:

\begin{lemma}\label{linear}
$A_{RQ}(\varepsilon)$ is a piecewise linear function.
\end{lemma}

\begin{proof}

Let $r_{j_k}$ be the $k^{th}$ rectangle generated by the $j^{th}$ segment of $Q$ as illustrated in Figure \ref{fig:area_events}. Since only the widths of the rectangles change in the scaling operation, we denote the base and the height of $r_{j_k}$ by $b_{j_k}(\varepsilon)$ and $h_{j_k}$, respectively.
According to the types of rectangles, the area function between two consecutive events can be expressed as:
\begin{equation*} \label{AreaRQ}
A_{RQ}(\varepsilon)=\sum_{r_{j_k}\in C_0}b_{j_k}(\varepsilon) h_{j_k}+\sum_{r_{j_k}\in C_1}b_{j_k}(\varepsilon) h_{j_k}+\sum_{r_{j_k}\in C_2}b_{j_k}(\varepsilon) h_{j_k}+\sum_{r_{j_k}\in C_3}b_{j_k}(\varepsilon) h_{j_k}.
\end{equation*}
Now let $\varepsilon' < \varepsilon''$ be two scaling moves between two consecutive events. Then
\begin{equation} \label{delta_dif}
   A_{RQ}(\varepsilon'') = A_{RQ}(\varepsilon') + (\varepsilon'' - \varepsilon')\left(\sum_{r_{j_k}\in C_1} h_{j_k}+\sum_{r_{j_k}\in C_2}j h_{j_k}-\sum_{r_{j_k}\in C_3}(j-1) h_{j_k}\right).
\end{equation}

Thus the area function is linear between two consecutive events.
\end{proof}

As a consequence of Lemma \ref{linear}, a solution of Problem 1; i.e., the minimum area in the scaling operation, is reached at the events and our approach is to efficiently sweep the events from left to right,
keeping the minimum value of area. The following result shows how to update the area value at each event.

\begin{lemma}\label{update_dif}
$A_{RQ}(\varepsilon)$ can be updated between two consecutive events in O(1) time.
\end{lemma}

\begin{proof}
Let $\varepsilon_i$, $\varepsilon_{i+1}$ be two consecutive events.
$A_{RQ}(\varepsilon)$ is a continuous piecewise linear function whose value at $\varepsilon_{i+1}$ can easily be computed from $A_{RQ}(\varepsilon_i)$ as follows.
Note that an event can be viewed in the abscissa axis as a swap between points $t_j$ and $x_j$ (see Figure \ref{fig:area_events}). Indeed, the points involved in the event are $t_j$, $x_i$ and their neighbors. Thus, the changes at any event are: a type $C_3$ rectangle disappears, a type $C_2$ appears and two neighbouring rectangles change their types. All the other rectangles maintain their type. Thus, Equation \eqref{delta_dif} can be used to compute the area at the new event $\varepsilon_{i+1}$ in $O(1)$ time. 
\end{proof}

We are ready to efficiently solve the problem.

\begin{theorem}
Problem \ref{prob:scaling-area} can be solved in $O(nm\log m$) time.
\end{theorem}

\begin{proof}
First, for every segment of the query $Q$, we store its events in an ordered list. Each element in the list holds the epsilon value, the height, and the type of the rectangles involved. We have $m$ lists overall. Since the input for $R$ is an ordered sequence of segments, each list can be computed in $O(n)$ time; thus we spend $O(nm)$ time to compute the $m$ lists.

Given the area for $\varepsilon=0$, $A_{RQ}(0)$, we build a min-heap whose nodes are $m$ events, the first of each list. Sweeping with a vertical line from left to right we use the lists to maintain the next event  in the heap
and update the area at each event. The process ends when there are no events in the  heap.
Since updating the area spends $O(1)$ time, by Lemma \ref{update_dif},  the sweeping algorithm spends $O(mn\log m)$ time.
\end{proof}

 \subsection{\emph{t}-Monotone Matching as Similarity Measure}

 Let $R=\{R_i=(x_i,p_i),\,  i=1,\dots,n\}$ and $Q=\{Q_j=(t_j,q_j),\,  j=1,\dots,m$\} be two melodies in point representation. For the sake of simplicity, we denote  the pair $(R_i ,Q_j)$ as $(i,j)$.

 Given $\varepsilon>0$, a $t$-monotone matching between $R$ and $S_Q(\varepsilon)$ can be partitioned into two sets $A^{-}(\varepsilon)$ and $A^{+}(\varepsilon)$ formed by the pairs of matched points $(i,j)$ where $x_i<t_j(\varepsilon)$ and $ x_i\geq t_j(\varepsilon)$, respectively.  The matching in Figure~\ref{fig:scaling}(b) gives rise to sets $A^{-}(\varepsilon)=\{(1,1),(3,2), (5,4), (6,4), (8,6)\}$ and $A^{+}(\varepsilon)=\{(1,2),(3,4), (4,7), (5,8)\}$.
We call the values of $\varepsilon>0$ where the sets $A^{-}$ and $A^{+}$
may change according the $t$-monotone matching definition \emph{events}.
There are three different types of events (see Figure \ref{fig:scaling_events}):
\begin{description}
	\item [Type 1.] $t_j(\varepsilon) = x_i$:  In this case,	the pair $(i,j)$ \emph{disappears} from $A^+$ and \emph{enters} in $A^-$.
	That is, for a small enough value $\delta\varepsilon$:
	\begin{align*}
		A^{+}(\varepsilon ) &= A^{+}(\varepsilon-\delta\varepsilon)\setminus\{(i,j)\}, \text{ and}\\
		A^{-}(\varepsilon) &= A(\varepsilon-\delta\varepsilon)\cup\{(i,j)\}.
	\end{align*}
	\item [Type 2.] $\frac{t_j(\varepsilon)+t_{j+1}(\varepsilon)}{2}=x_i$:  The pair $(i,j+1)$ \emph{disappears} from $A^-$ and $(i,j)$ \emph{enters} in $A^+$.
	That is,
	\begin{align*}
		A^{+}(\varepsilon) &= A^{+}(\varepsilon-\delta\varepsilon)\cup\{(i,j)\}, \text{ and}\\
		A^{-}(\varepsilon) &= A^{-}(\varepsilon - \delta\varepsilon)\setminus\{(i,j+1)\}.
	\end{align*}

	In addition to the changes in $ A^{-} $ and $ A^{+} $, we point out the following:

	\begin{itemize}
		\item If
		$Q_j$ was an unmatched point in $S_Q(\varepsilon-\delta\varepsilon)$; that is, it
		was not  the $t$-monotone nearest point from any $R$,
		then it was
		matched with its $t$-monotone nearest point among $R_{i-1}$ and $R_{i}$.
		If $Q_j$ was matched with $R_i$, then $(i,j)$ is already in $A^+(\varepsilon)$ (i.e.,  $A^+(\varepsilon)=A^{+}(\varepsilon-\delta\varepsilon)$).
		If $Q_j$ was matched with $R_{i-1}$ (i.e., $(i-1,j)\in A^{-}(\varepsilon - \delta\varepsilon)$), then when $(i,j)$ \emph{enters} $A^+(\varepsilon)$ the pair $(i-1,j)$ must \emph{disappear} from $A^-(\varepsilon)$.
		\item Also, note that $Q_{j+1}$ may become unmatched when $(i,j+1)$ \emph{disappears} from $A^{-}(\varepsilon)$ (i.e., $x_{i+1}>t_{j+1}(\varepsilon+\delta\varepsilon)$ and $Q_{j+2}$ is $t$-monotone nearer to $R_{i+1}$ than $Q_{j+1}$). In this case $Q_{j+1}$ must be matched with its $t$-monotone nearest point, $R_i$ or $R_{i+1}$. If $R_i$ is $t$-monotone nearer to $Q_{j+1}$ than $R_{i+1}$, then $(i,j+1)$ must remain in $A^-(\varepsilon)$. Otherwise, $(i+1,j+1)$ must enter $A^+$.
		
	\end{itemize}

	\item [Type 3.] $t_j(\varepsilon)=\frac{x_i+x_{i+1}}{2}$ and in $S_Q(\varepsilon)$, $Q_j$ is not the nearest $t$-coordinate point from any $R$ note (i.e., $Q_{j-1}$ and $Q_{j+1}$ are $t$-monotone nearer to $R_i$ and $R_{i+1}$, respectively, than $Q_j$):
	The pair $(i,j)$ \emph{disappears} from $A^-(\varepsilon)$ and $(i+1,j)$ \emph{enters} in $A^+(\varepsilon)$.
	That is,
	\begin{align*}
		A^{+}(\varepsilon) &= A^{+}(\varepsilon-\delta\varepsilon)\cup\{(i+1,j)\}, \text{ and}\\
		A^{-}(\varepsilon) &= A^{-}(\varepsilon-\Delta\varepsilon)\setminus\{(i,j)\}.
	\end{align*}
\end{description}

\begin{figure}[h!]
	\centering
	\includegraphics[scale=0.8]{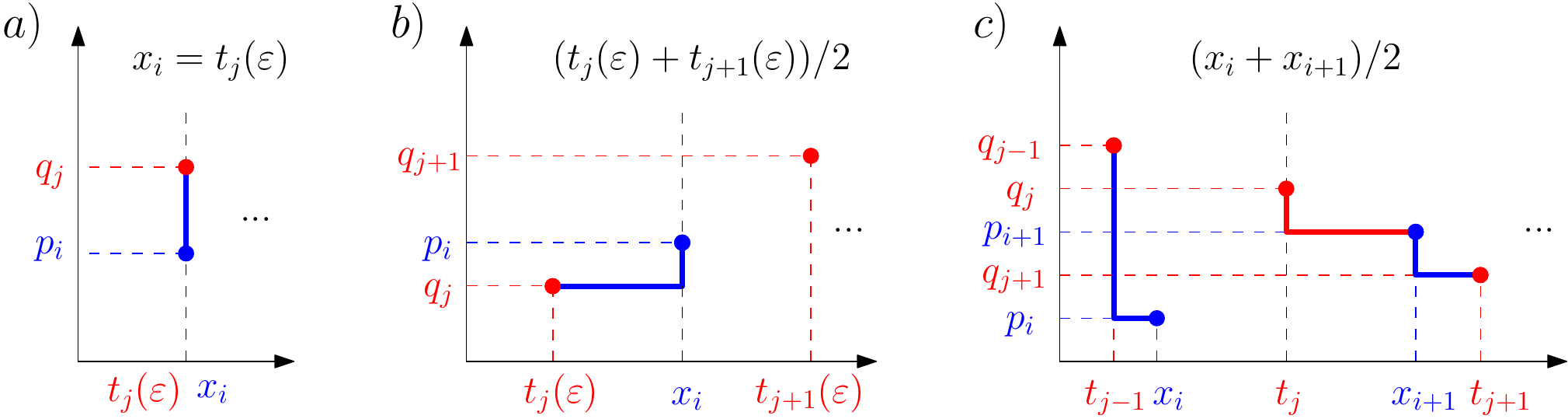}
	\caption{Events of a) type 1, b) type 2, and  c) type 3. }
	\label{fig:scaling_events}
\end{figure}

Note that when no events occur for a small enough value $\delta\varepsilon$, the set of pairs does not change; i.e., $A^{-}(\varepsilon-\delta\varepsilon)=A^{-}(\varepsilon+\delta\varepsilon)$ and $A^{+}(\varepsilon-\delta\varepsilon)=A^{+}(\varepsilon+\delta\varepsilon)$.
Thus, we denote the sets $A^{-}$ and $A^{+}$  between two consecutive events $\varepsilon_{\ell}$ and $\varepsilon_{\ell+1}$,    $A^{-}_\ell$ (resp.\ $A^{+}_\ell$) the set  $A^{-}(\varepsilon)$ (resp.\ $A^{+}(\varepsilon)$), for all $\varepsilon\in(\varepsilon_\ell,\varepsilon_{\ell+1})$.
Moreover, the costs related to pairs in $A_{\ell}^{-}$ (resp.\ $A_{\ell}^{+}$) increase  (resp.\ decrease).

 \begin{lemma}\label{lem:events}
 The ordered sequence of events $\varepsilon_1<\varepsilon_2<\dots<\varepsilon_{\ell}<\varepsilon_{\ell+1}<\cdots$ with the  pairs of disappearing and entering matched points involved can be computed in $O(nm\log m)$ time.
 \end{lemma}
 \begin{proof}

 For every note $Q_j$ in $Q$ we can sweep $R$ from left to right with a vertical line starting at $t_j$ and compute an ordered sequence $\varepsilon_1<\varepsilon_2<\dots$ of scaling events, such that $t_j(\varepsilon_\ell)=x_i$ (type 1 event) or $t_j(\varepsilon_\ell)=(x_i+x_{i+1})/2$ (type 3 event). Notice that for each $\varepsilon_\ell$ scaling event detected when sweeping $R$, we can compute the set of \emph{disappearing} and \emph{entering} pairs by analysing a constant number of notes (of $R$ and $S_Q(\varepsilon_\ell)$) around $t_j(\varepsilon_\ell)$. This processing takes $O(n)$ for each note in $Q$. In total we get $m$ ordered sequences (one per note in $Q$) in $O(nm)$ time.

 Analogously to this sweeping, for every vertical bisector between two notes $Q_j$ and $Q_{j+1}$, we can sweep $R$ from left to right with a vertical line starting at $(t_j+t_{j+1})/2$ and compute an ordered sequence $\varepsilon_1<\varepsilon_2<\dots$ of scaling events such that $(t_j(\varepsilon_\ell)+t_{j+1}(\varepsilon_\ell))/2=x_i$ (type 2 event). As before, we can compute the set of \emph{disappearing} and \emph{entering} pairs for every $\varepsilon_\ell$ in constant time.
 This processing takes $O(n)$ for each pair of consecutive notes in $Q$. In total we get $m-1$ ordered sequences in $O(nm)$ time.

After all the sweeping processes are completed, we have $O(m)$ ordered sequences of length $O(n)$. We can merge all these sequences into a single ordered sequence of events in $O(nm\log m)$ time. 
 \end{proof}

 \begin{lemma}\label{lem:linear_func}
 The cost ${\phi}(R,S_Q({\varepsilon}))$ is a linear function between two consecutive events.
 \end{lemma}
\begin{proof} Let $\varepsilon_\ell<\varepsilon<\varepsilon'<\varepsilon_{\ell+1}$, where $\varepsilon_\ell, \varepsilon_{\ell+1}$ are consecutive events.
\begin{equation*} \label{cost}
 {\phi}(R,S_Q({\varepsilon}))=\sum_{(i,j)\in A^{-}_\ell}(t_j(\varepsilon)-x_i)+\sum_{(i,j)\in A^{+}_\ell}(x_i-t_j(\varepsilon))+\sum_{(i,j)\in A^{-}_\ell\cup A^{+}_\ell}\left|p_i - q_j\right| .
 \end{equation*}
 Now, using that $t_j(\varepsilon')=t_j(\varepsilon)+(2j-1)(\varepsilon'-\varepsilon)/2$, we get that
 \begin{equation}\label{formupdate}
{\phi}(R,S_Q({\varepsilon'}))=
{\phi}(R,S_Q({\varepsilon})) + \dfrac{(\varepsilon' - \varepsilon)}{2} \left(\sum_{(i,j)\in A^{-}_\ell}(2j-1) - \sum_{(i,j)\in A^{+}_\ell} (2j-1) \right),
\end{equation}
and the result follows.
\end{proof}

We are ready to introduce the main result of this section.

\begin{theorem}\label{thm:main_scaling_points}
    Problem \ref{prob:scaling-matching} can be solved in $O(nm\log m)$ time.
\end{theorem}

\begin{proof}
     From Lemma~\ref{lem:linear_func}, we have that the minimum value for ${\phi}(R,S_Q({\varepsilon}))$ is attained at some value of the sequence $\varepsilon_0=0\leq\varepsilon_1<\varepsilon_2<\dots<\varepsilon_{\ell}<\varepsilon_{\ell+1}<\cdots$, where $\varepsilon_0$ represents the initial $t$-monotone matching between $R$ and $Q$. Note that the piecewise linear function $\phi$ can be discontinuous at the events and the solution is then found by evaluating the cost at the endpoints of each piece. See Figure \ref{fig:phi} for an illustration. At events of type 1, the cost function is continuous but for type 2 and type 3, it may not be.

\begin{figure}[h!]
\centering
\includegraphics[width=3in]{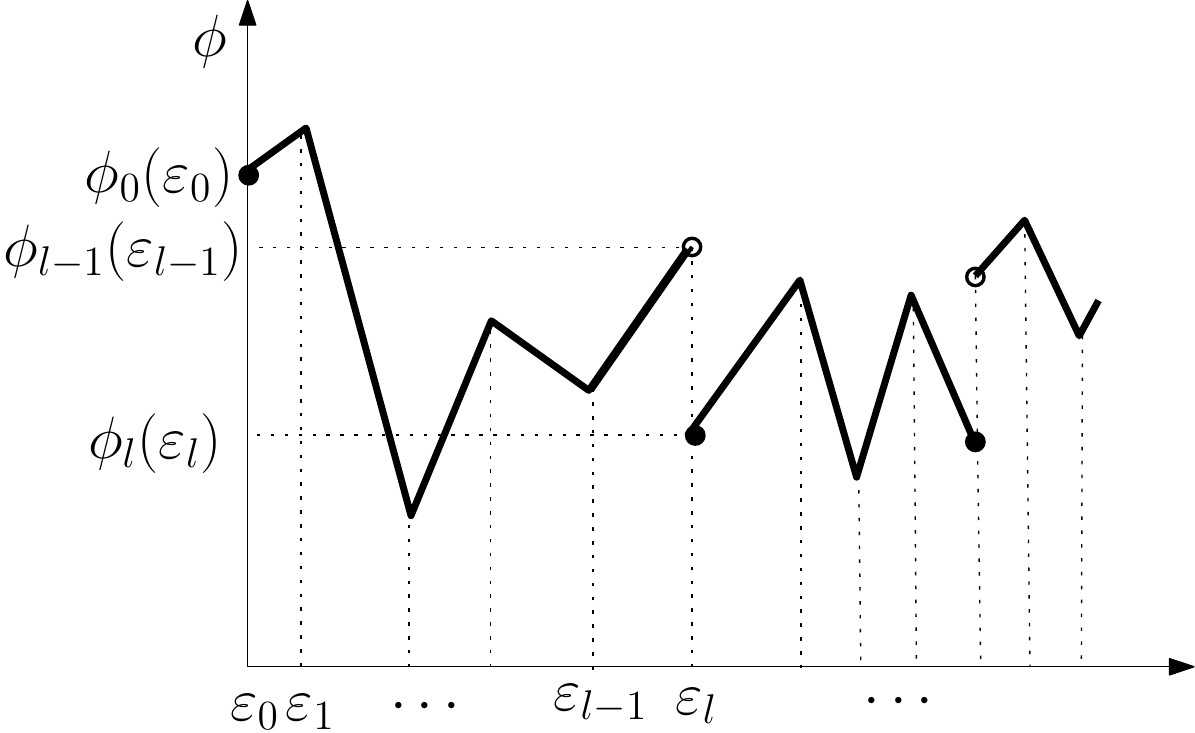}
\caption{$\phi$ is continuous at $\varepsilon_1$  (type 1 event). Later,
the function jumps because of the $l_1$ metric assignment at event $\varepsilon_\ell$ (type 2 or type 3). }
\label{fig:phi}
\end{figure}

The algorithm is as follows. First, compute the sequence of events $\varepsilon_1<\varepsilon_2<\dots<\varepsilon_\ell<\dots$ with the set of \emph{entering} and \emph{disappearing} pairs in each event (this can be done in $O(nm\log m)$ time; Lemma~\ref{lem:events}). We know that $\phi(R, S_Q(\varepsilon))$ is a piecewise linear function (Lemma~\ref{lem:linear_func}), so we conveniently denote by $\phi_\ell(\varepsilon)$ the piece of $\phi(R, S_Q(\varepsilon))$ between $\varepsilon_\ell$ and $\varepsilon_{\ell+1}$ for all $\ell\geq 0$ (see Figure \ref{fig:phi}). Note that
\begin{equation}\label{eq:phi_at_events}
\phi(R, S_Q(\varepsilon))=\left\{
\begin{array}{ll}
   \phi_\ell(\varepsilon) & \text{if\quad}\varepsilon_\ell<\varepsilon<\varepsilon_{\ell+1},\\
   \min\{\phi_{\ell-1}(\varepsilon_\ell),\phi_\ell(\varepsilon_\ell)\} & \text{if\quad}\varepsilon=\varepsilon_\ell.
\end{array}
\right.
\end{equation}

Now we elaborate on how to progressively compute the values \[\phi_0(\varepsilon_0),  \phi_0(\varepsilon_1),  \phi_1(\varepsilon_1),  \phi_1(\varepsilon_2),\ldots,  \phi_\ell(\varepsilon_\ell),  \phi_\ell(\varepsilon_{\ell+1}),\ldots\]

Let $S^{-}_\ell=\sum_{(i,j)\in A^{-}_\ell}(2j-1)$ and $S^{+}_\ell=\sum_{(i,j)\in A^{+}_\ell}(2j-1)$ (the sum of indices of  (\ref{formupdate})).

At the initial step, from the $t$-monotone matching between $R$ and $Q$, we compute the values $\phi_0(\varepsilon_0)$, $S^{-}_0$ and $S^{+}_0$ in $O(n+m)$ time.
For $\ell=1,2,\dots$ we proceed as follows:
\begin{enumerate}
    \item $\phi_{\ell-1}(\varepsilon_\ell) = \phi_{\ell-1}(\varepsilon_{\ell-1})+\frac{\varepsilon_\ell-\varepsilon_{\ell-1}}{2}\left(S^{-}_{\ell-1}-S^{+}_{\ell-1}\right)$ (Lemma~\ref{lem:linear_func})
    \item $\phi_\ell(\varepsilon_\ell) = \phi_{\ell-1}(\varepsilon_\ell) - D + E$ where $D$  denotes the sum of the $l_1$-distance of the \emph{disappearing} pairs of points from $A^-$ and $A^+$. Analogously, $E$ denotes the sum of the $l_1$-distance of the \emph{entering} pairs of points in $A^-$ and $A^+$.
    \item $S^{-}_{\ell}$ and $S^{+}_{\ell}$ can be easily computed from $S^{-}_{\ell-1}$ and $S^{+}_{\ell-1}$ because we have the \emph{disappearing} and \emph{entering} pairs in $\varepsilon_\ell$ (Lemma \ref{lem:events}).
\end{enumerate}
Note that $\phi$ is continuous in $\varepsilon_\ell$ (i.e., $\phi_\ell(\varepsilon_\ell) = \phi_{\ell-1}(\varepsilon_\ell)$) if and only if $D=E$. For example, $D=E$ at events of type 1 because the disappearing point in $A^+$ is an entering point in $A^-$.

While we are sweeping the list of events, we can compute the values $\phi(R,S_Q(\varepsilon_\ell))$ by using equation \eqref{eq:phi_at_events}, and by keeping the overall minimum evaluation, we get the best scaling at the end of this process. This sweeping takes $O(nm)$ time (length of the list of events); therefore, the overall time cost of the algorithm is determined by the computation of the ordered list of events which takes $O(nm\log m)$ time (Lemma \ref{lem:events}).
 \end{proof}

\section{Compressing Operation}\label{sec:comp}
	
\subsection{t-Monotone Matching as Similarity Measure}
	
Let $R=\{R_i=(x_i,p_i),\, i=1,\dots,n\}$ be a sequence of notes in point representation and $k\in\{1,\dots,n\}$.
Let $C_k(R)=\{R_{i_1},\dots, R_{i_k}\}\subseteq R$ be a $k$-compressed melody of $R$. We will call this subset a \emph{$k$-set of $R$} and we denote it by $C_k$.
Consider a $t$-monotone matching between $R$ and $C_k$ with cost  $\phi(R,C_k)=\sum_{1\leq j \leq k}\phi(R_{i_j})$.
A solution of Problem \ref{prob:compressing-matching} is an optimal $k$-set of $R$ minimizing the cost ${\phi}$.

\begin{definition}
Let $C_k=\{R_{i_1},\dots,R_{i_k}\}$ be a $k$-set of $R$.
\begin{enumerate}
\item Set  $\overleftarrow{R}(C_k)=\{R_i\in R\colon x_i\leq x_{i_k})\}$. The \emph{left partial evaluation} of $C_k$ is defined as
\[\overleftarrow{\phi}(R,C_k)=\phi(\overleftarrow{R}(C_k),C_k).\]
\item We say that $C_k$ is \emph{left optimal} if $\overleftarrow{\phi}( R,  C_k)\leq \overleftarrow{\phi}(R,  C_k')$ for all $C_k'=\{R'_{i_1},\dots,R'_{i_k}\}$ where $R'_{i_k}=R_{i_k}$.

\item The  \emph{$j$-prefix of $C_k$} is the $j$-set formed by the first $j$ points of $C_k$.
\end{enumerate}
\end{definition}

 The following result can easily be proven by contradiction.
	\begin{lemma}\label{lem:prefix}
		Let $C_k^*=\{R^*_{i_1},\dots,R^*_{i_k}\}$ be an optimal $k$-set of $R$. Then for all $1\leq j \leq k$, the $j$-prefix of $C_k^*$ is left optimal.
	\end{lemma}
	
    This lemma can be directly extended as the following.
	
	\begin{corollary}\label{cor:prefix}
	Let $C_j=\{R_{i_1},\dots,R_{i_j}\}$ be a $j$-set of $R$ which is left optimal. Then every prefix of $C_j$ is also left optimal.
	\end{corollary}
	
These properties enable us to solve Problem \ref{prob:compressing-matching} with dynamic programming.
Let $R_i$ and $R_{i'}$ be two consecutive points in a $k$-set. Let $W_{i,i'}$ be the cost of the $t$-monotone matching for points between $R_i$ and $R_{i'}$. Then for a $k$-set $C_k=\{R_{i_1},\dots,R_{i_k}\}$, $1\leq i_j\leq n$ and $i_j<i_{j+1}$, the cost function $\phi(R,C_k)$ can be rewritten as $W_{0,i_1}+W_{i_1,i_2}+\ldots+W_{i_k,(n+1)}$, where $W_{0,i_1}$ (resp.\ $W_{i_k,(n+1)}$) denotes the assignment cost of the points to the left (resp.\ right) of $R_{i_1}$ (resp.\ $R_{i_k}$).
	
Now assume that we have preprocessed the values $W_{i,i'}$ for all $0\leq i<i'\leq n+1$.
Later we will show how to compute the values of $W_{i,i'}$ efficiently.
Consider the tables $C[i,j]$ and $P[i,j]$ whose keys $i$ and $j$ are integers in the intervals $[1,n]$ and $[1,k]$. The cell $C[i,j]$ (resp.\ $P[i,j]$) stores the cost (resp.\ the index) of the left optimal $j$-set $C_j$ that ends using $p_i$ as the $j$-th point of the subset.
Let see how to fill $C[i,j]$ and $P[i,j]$.
	
From the definitions of the tables $C$ and $P$, 
Lemma \ref{lem:prefix} and Corollary \ref{cor:prefix}, we derive the following dynamic programming formula.

	\begin{lemma}\label{lem:tables}
	 The values of the cells of tables $C$ and $P$ are:
	 \begin{align*}
	     C[i,1] & = W_{0,i},\\
	     P[i,1] & = 0,\\
	     C[i,j] & = \min_{1\leq l<i}\{C[l,j-1]+W_{l,i}\}=C[i^*,j-1]+W_{i^*,i},\\
	      P[i,j] & = i^*.
	 \end{align*}
	\end{lemma}

	\begin{theorem}\label{thm:dynamic}
	    Assuming that the values $W_{i,i'}$ are already known, the optimum $k$-set can be computed in $O(kn^2)$ time.
	\end{theorem}
	\begin{proof}
	    From Lemma \ref{lem:tables}, the tables $C$ and $P$ can be filled properly in $O(kn^2)$ time. Then the cost of the optimum $k$-set is
   	\[\min_{1\leq i\leq n}\{C[i,k]+W_{i,(n+1)}\},\]
    	and the elements of the $k$-set can be recovered navigating backward in the table $P$ (i.e., the $k$-th element is the point $R_{i^*}$ where $i^*$ minimizes the value $C[i^*,k]+W_{i^*,(n+1)}$, the $(k-1)$-th element is stored in $P[i^*,k]$, the $(k-2)$-th element is stored in $P[P[i^*,k],k-1]$, etc.). Computing the value $i^*$ takes $O(n)$ time and the process of  backward navigation in $P$ takes $O(k)$ time.
	\end{proof}

	Let us focus now on how to compute the values $W_{i,i'}$. Consider that $R_i$ and $R_{i'}$ are two consecutive selected elements ($i<i'$). Let $m=(x(R_i)+x(R_{i'}))/2$, and denote the Manhattan distance between two points by $d_1(\cdot)$. Consider
	        \[\mathcal{L}_{i,i'}=M_l+\sum_{\substack{
	R_j\in R,\\
	x(R_{i})<x(R_j)< m}}
	d_1(R_j,R_{i})\quad \text{and}\quad \mathcal{R}_{i,i'}=M_r+\sum_{\substack{
	R_{i'}\in R,\\ m< x(R_j)<x(R_{i'})}}
	d_1(R_j,R_{i'}),\]
	where $M_l=d_1(R^*,R_{i})$ if there is $R^*\in R$ such that $x(R^*)=m$ and $d_1(R^*,R_{i})\leq d_1(R^*,R_{i'})$; otherwise, $M_l=0$. Analogously, $M_r=d_1(R^*,R_{i'})$ if there is $R^*\in R$ such that $x(R^*)=m$ and $d_1(R^*,R_{i'})< d_1(R^*,R_{i})$; otherwise $M_r=0$. Notice that $M_l$ and $M_r$ are not both greater than zero; in this way we are guaranteeing that the $x$-middle point between $R_{i}$ and $R_{i'}$ (if it exists) is assigned to only one center.
	Note that $W_{i,i'}=\mathcal{L}_{i,i'}+ \mathcal{R}_{i,i'}$.  For convenience, we consider that $\mathcal{L}_{0,i}=\mathcal{R}_{i,(n+1)}=0$ for all $1\leq i\leq n$. Also, $\mathcal{R}_{0,i}$ (resp.\ $\mathcal{L}_{i,(n+1)}$) is the assignment cost of the points to the left (resp.\ right) of $R_{i}$.
	
	\begin{lemma}\label{lem:swapping}
	 For a given fixed value $i$, $\mathcal{L}_{i,i+1}, \mathcal{L}_{i,i+2},\dots,\mathcal{L}_{i,(n+1)}$ (resp.\ $\mathcal{R}_{i-1,i}, \mathcal{R}_{i-2,i},\dots,\mathcal{R}_{0, i}$) can be computed in $O(n)$ time.
	\end{lemma}
	
	\begin{proof}
	Using a vertical line to make a left--right sweep from $R_i$, we can compute all the values $\mathcal{L}_{i,i+1}, \mathcal{L}_{i,i+2},\dots, \mathcal{L}_{i,(n+1)}$ in $O(n)$ time. Analogously, we can compute the sequence of values $\mathcal{R}_{i-1,i}$, $\mathcal{R}_{i-2,i}$, \dots, $\mathcal{R}_{0, i}$ by using a right--left  sweep
	 from $R_i$.
	\end{proof}
	
	\begin{lemma}\label{lem:joining}
All values $W_{i,i'}$ can be computed in $O(n^2)$ time.
	\end{lemma}
	\begin{proof}
	All the sequences $\mathcal{L}_{i,i+1},\dots,\mathcal{L}_{i,(n+1)}$ and $\mathcal{R}_{i-1,i}, \mathcal{R}_{i-2,i},\dots,\mathcal{R}_{0, i}$ can be computed in $O(n^2)$ time (there are $O(n)$ sequences and computing the values of a sequence takes $O(n)$ time by Lemma~\ref{lem:swapping}). Then computing the value $W_{i,i'}=\mathcal{L}_{i,i'}+\mathcal{R}_{i,i'}$ for all $i<i'$ takes $O(n^2)$ time.
	\end{proof}
	
	From Lemma \ref{lem:joining} and Theorem~\ref{thm:dynamic} we arrive at the following theorem.
	
	\begin{theorem}
Problem \ref{prob:compressing-matching} can be solved in $O(kn^2)$ time.
	\end{theorem}
	
\subsection{Area as Similarity Measure}

Given a melody $R$ with $n$ notes  in the segment representation  and a positive integer  $k<n$, we want to find a $k$-compressed melody $C_k(R)$ of $R$ that minimizes the area between $R$ and $C_k(R)$.
Let $P=\{p_1,\dots, p_\rho \}$, $\rho\leq n$, be the set of different pitch values among the notes of $R$ (recall that two non-consecutive notes may have the same pitch).
Let $p(x_i)\in P$ denote the pitch of the note in $R$ that starts at time $x_i$ ($0\leq i<n$).
In Figure \ref{fig:compressing} b), a solution of Problem \ref{prob:compressing-area} for $k=2$ is shown for which the two supporting segments are $p(x_1)$ and $p(x_6)$.
In this section, we propose an algorithm to solve Problem \ref{prob:compressing-area} in $O(k\rho n)$  time, which is subquadratic in many cases in oral music traditions where $\rho$ is much smaller than $n$.

\begin{lemma}\label{lem:technical}
Given a melody $R$
with $n$ notes, in segment representation, let $X = (x_0=0,\dots,x_n)$ be the time partition of $R$. For every  $k < n$, there exists a melody  $M^*_k$ with time partition $T=(t_0=0,\dots, t_k=x_n)$ that minimizes the difference in area from $R$ over all all the melodies formed by $k$ notes starting at time $x_0$ and ending at time $x_n$ such that $M^*_k$ fulfills the following two properties:
\begin{enumerate}
    \item $T\subseteq X$, and
    \item each segment of $M^*_k$ contains at least a segment of $R$.
\end{enumerate}
\end{lemma}
\begin{proof}
Let $\mathcal{M}_k$ be the set of all the melodies formed by $k$ notes starting at time $x_0$ and ending at time $x_n$.
We first prove that there exists a melody in $\mathcal{M}_k$ with time partition $T\subseteq X$ that minimizes the difference in area from $R$. For the sake of contradiction, suppose that no such melody exists. Let $M_k\in \mathcal{M}_k$ be the melody that minimizes the difference in area from $R$ such that $\left|T\setminus X\right|$ is minimal. Due to our assumption, $\left|T\setminus X\right|>0$. Then there is a value $0<j<k$ and a value $0\leq i<n$ such that $x_i<t_j<x_{i+1}$; see Figure \ref{fig:area_diff}. It is easy to check that shifting $t_j$ toward $x_i$ or toward $x_{i+1}$ we get a melody $M_k'\in \mathcal{M}_k$ whose difference in area with respect to $R$ is less than or equal to the difference in area between $R$ and $M_k$. Moreover, denoting by $T'$ the time partition of $M_k'$, we have that $\left|T'\setminus X\right|<\left|T\setminus X\right|$, a contradiction.

\begin{figure}
    \centering
    \includegraphics[scale=.7]{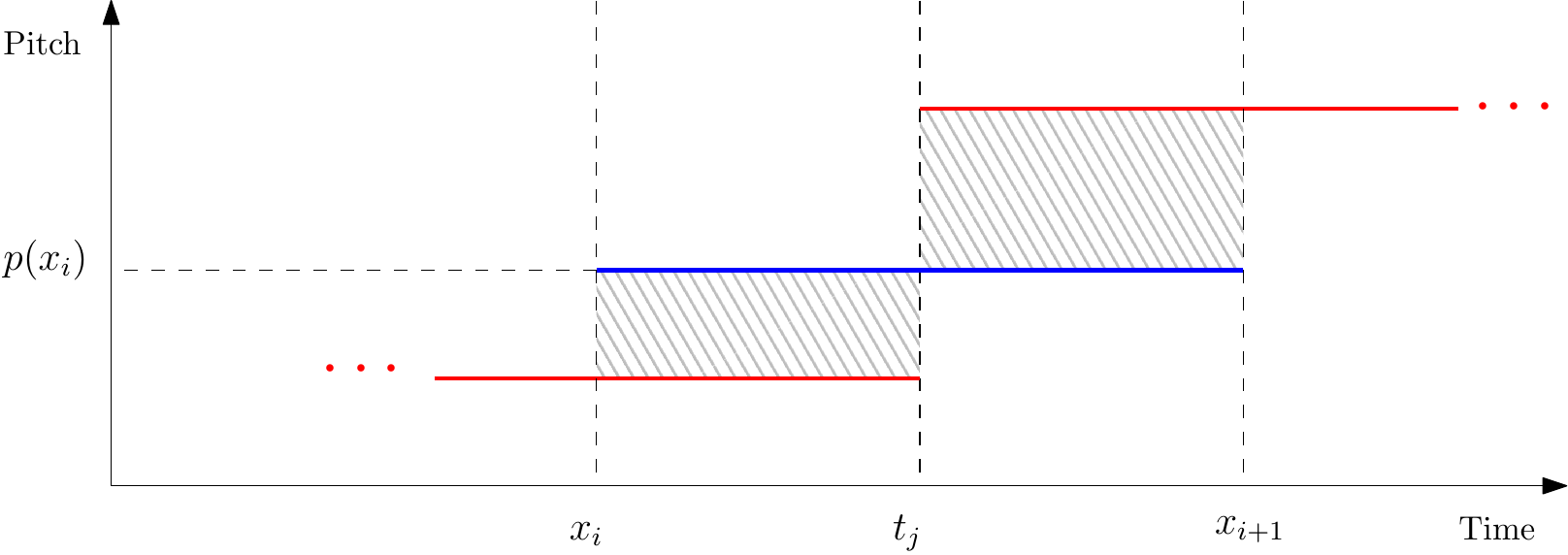}
    \caption{Illustration of Lemma \ref{lem:technical}. By shifting $t_j$, toward $x_{i+1}$ the area between the melodies is decreased.}
    \label{fig:area_diff}
\end{figure}

Let $\mathcal{M}^*_k\subseteq \mathcal{M}_k$ be the set of all the melodies whose time partition is a subset of $X$ and for which the difference in area with $R$ is minimum. We now prove that there exists a melody in $\mathcal{M}^*_k$ in which each note contains a note of $R$. For the sake of contradiction, suppose that no such melody exists. Let $M_k$ be a melody in $\mathcal{M}^*_k$ with a maximum number $q$ of notes containing at least one note of $R$.  Due to our assumption, $k-q>0$, and $k-q$ is minimal. Suppose that the segment $Y$  of $M_k$ between times $t_i$ and $t_{i+1}$ does not contain a note of $R$. It is easy to check that lowering or raising $Y$ until the first segment  of $R$ in the interval $[t_i,t_{i+1}]$ is below or above $Y$, respectively, we obtain a melody $M_k'$ whose difference in area from $R$ is at most $A$ and $q+1$ notes of $M_k'$ contain at least one note of $R$, a contradiction.
\end{proof}

From the above lemma, the following result is directly deduced.
\begin{corollary}\label{cor:j-compression}
Let $R$ be a melody  in  segment representation with time partition $X=(x_0=0,\dots,x_n)$. For every  $k < n$, there exists a $k$-compressed melody of $R$ with time partition $T=(t_0=x_0,\dots,t_k=x_n)$ that minimizes the difference in area from $R$ over all the melodies formed by $k$ notes starting at time $x_0$ and ending at time $x_n$.
\end{corollary}

The above results justify the definition of a $k$-compressed melody and allow us to look for a solution of Problem \ref{prob:compressing-area} that fulfills properties 1 and 2 of Lemma~\ref{lem:technical}.

\begin{definition} Let $R$ be a melody  with time partition $X=(x_0,\dots, x_n)$. For every $x_i\in X$,  we denote by $R_{\overleftarrow{x_i}}$  the \emph{prefix melody} of $R$ with time partition $X_{\overleftarrow{x_i}}=(x_0,\dots,x_i)$.

\end{definition}

\begin{definition} \label{def:sp}
 Let $R$ be a melody  with time partition $X=(x_0,\dots, x_n)$ and let $P$ be the set of different pitch values among the notes of $R$. Fixing  $k<n$ and three values $1\leq j\leq k$, $j\leq i\leq n$ and $p\in P$, let $S_i[j,p]$ be the set of all the melodies formed by $j$ segments starting at time $x_0$ and ending at time $x_i$ whose last segment has pitch $p\in P$ and starts at some time $x_{i'}\in X$ with $i'<i$. See Figure~\ref{fig:sp-2} for an illustration.
\end{definition}

\begin{figure}[ht]
\centering
\includegraphics[scale=0.6]{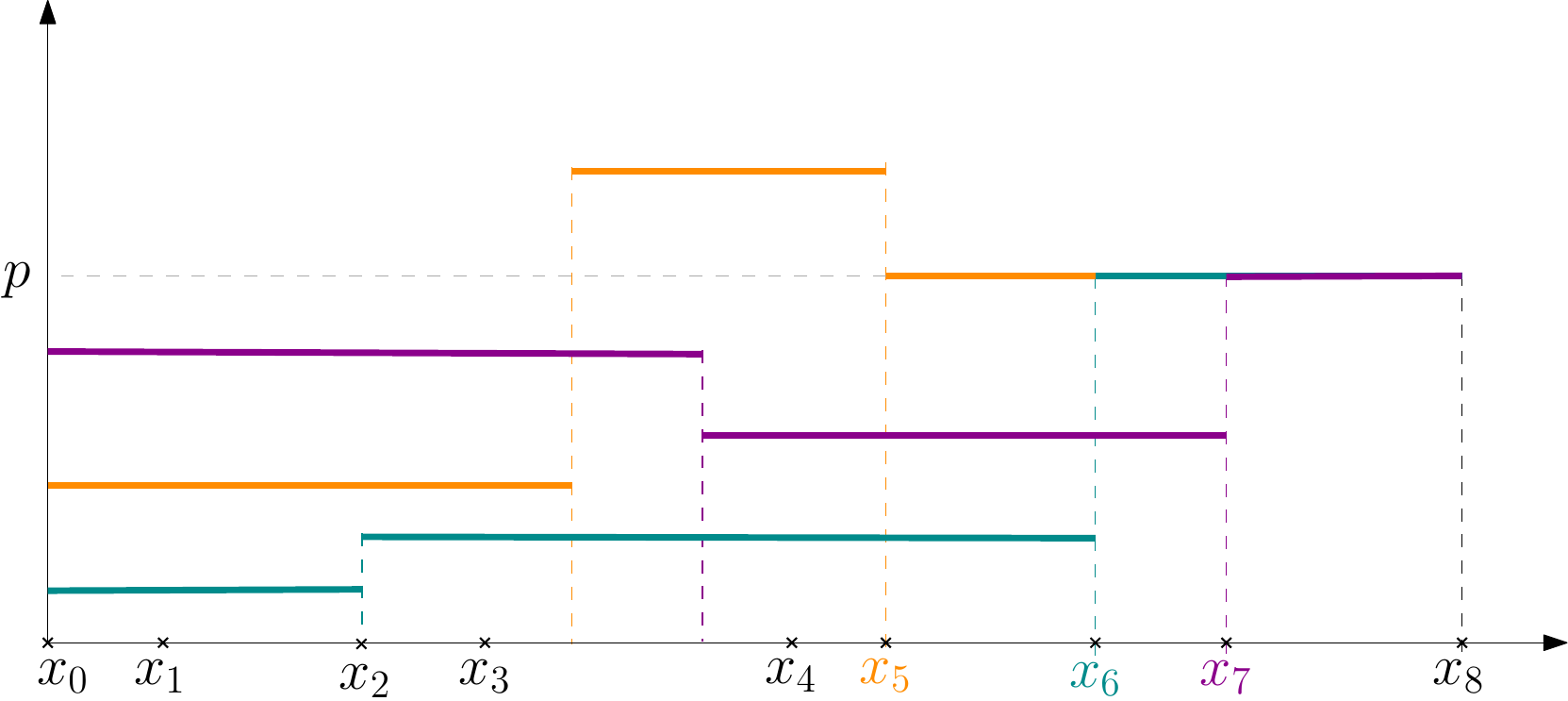}
\caption{Green, purple and orange melodic contours are in $S_8[3,p]$.}
\label{fig:sp-2}
\end{figure}

Note that a  $k$-compressed melody, $C_k(R)$  whose last note has pitch $\hat{p} \in P$ belongs to  $S_n[k,\hat{p}]$.

 In the following we propose a dynamic programming approach in order to compute an optimum $k$-compression of $R$.
The next result establishes an optimal substructure of the $k$-compressed melodies.

\begin{lemma}\label{lem:compression-prefix}
There exists a melody $C\in S_i[j,p]$ that minimizes the difference in area from $R_{\overleftarrow{x_i}}$ over all the melodies in $S_i[j,p]$, and
whose prefix $C_{\overleftarrow{x_{i'}}}$ is an optimum $(j-1)$-compressed melody of $R_{\overleftarrow{x_{i'}}}$, where $x_{i'}$ is the starting point of the last segment of $C$.
\end{lemma}
\begin{proof}
Let $C'\in S_i[j,p]$ be a melody that minimizes the difference in area from $R_{\overleftarrow{x_i}}$ over all the melodies in $S_i[j,p]$.
If $C'_{\overleftarrow{x_{i'}}}$ is an optimum $(j-1)$-compression of $R_{\overleftarrow{x_{i'}}}$ then take $C=C'$ and we are done. Otherwise, by Corollary \ref{cor:j-compression}, there exists a $(j-1)$-compression $C_{j-1}$ of $R_{\overleftarrow{x_{i'}}}$ that minimizes the difference. Using $C_{j-1}$ plus the last segment of $C'$, we obtain a melody $C''$ which is also in $S_i[j,p]$ and with minimum difference in area; thus take $C=C''$.
\end{proof}

 Figure~\ref{fig:lema24} shows an example to illustrate Lemma~\ref{lem:compression-prefix}.

\begin{figure}[ht]
\centering
\includegraphics[scale=0.7]{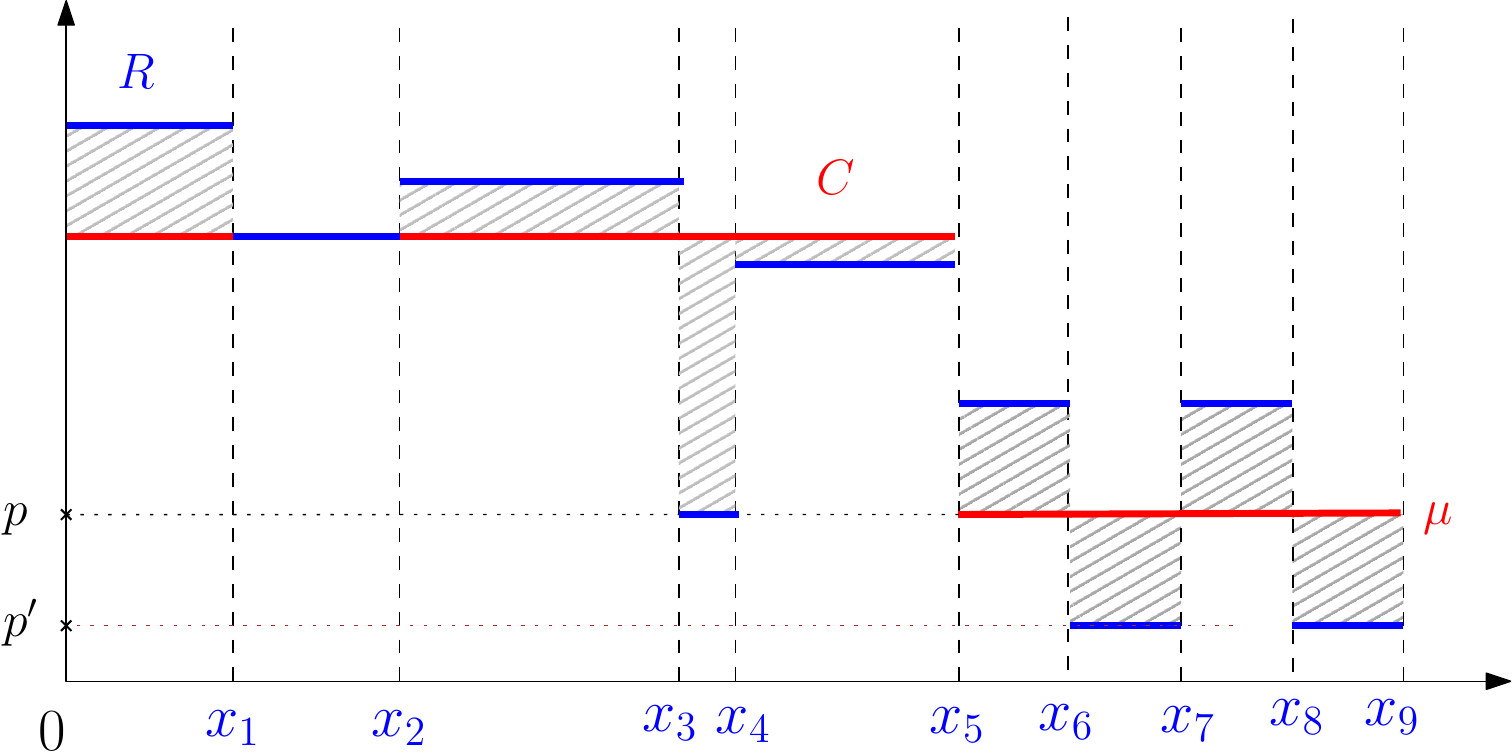}
\caption{The melody $C$ belongs to $S_{9}[2,p]$ and it is not a $2$-compression of $R$ because its last segment does not contain any segment of $R$. The prefix $C_{\protect\overleftarrow{x_5}}$ is an optimum 1-compression of $R_{\protect\overleftarrow{x_5}}$.}
\label{fig:lema24}
\end{figure}

\begin{definition}\label{def:A,T}
Let $C\in S_i[j,p]$ be a melody that minimizes the difference in area from $R_{\overleftarrow{x_i}}$ over all the melodies in $S_i[j,p]$.
We denote by
$A_i[j,p]$ the difference in area between $R_{\overleftarrow{x_i}}$ and $C$ and by $T_i[j,p]=i'$ the index of the starting time of the last segment in $C$.
\end{definition}

\begin{lemma}\label{lem:dp-1}
Let $p^*\in P$ be the minimum pitch value such that $A_i[j,p^*]=\displaystyle\min_{  p\in P}A_i[j,p]$. There  exists a melody $C\in S_i[j,p^*]$ which is an optimal $j$-compression of $R_{\overleftarrow{x_i}}$.
\end{lemma}
\begin{proof}
From Lemma \ref{lem:technical}, Definition \ref{def:A,T} and the assumptions on $p^*$, it is easy to deduce that the difference in area between $R_{\overleftarrow{x_i}}$ and its optimum $j$-compression is $A_i[j,p^*]$.

Let $C\in S_i[j,p^*]$ be a melody whose difference in area from $R_{\overleftarrow{x_i}}$ is $A_i[j,p^*]$. Let $\mu$ denote the last segment of $C$. If $\mu$ contains some segment of $R_{\overleftarrow{x_i}}$ then $C$ is a $j$-compression and we are done.  For the sake of contradiction, assume that $\mu$ does not contain any segment of $R_{\overleftarrow{x_i}}$.
Let $L_b$ and $L_a$ be the sum of the lengths of the segments below and above $\mu$ respectively.

If $L_b\geq L_a$, then lowering $\mu$ to the segment immediately below, we get a new melody $C'$ whose last note has a pitch value less than $p^*$ and the difference in area between $C'$ and $R_{\overleftarrow{x_i}}$ is at most $A_i[j,p^*]$. This contradicts that $p^*$ is minimum.

If $L_b< L_a$, then raising $\mu$ to the segment immediately above, we get a new melody $C'$ whose difference in area with respect to $R_{\overleftarrow{x_i}}$ is less than $A_i[j,p^*]$. This contradicts that $A_i[j,p^*]$ is minimal and the result follows.
\end{proof}

\begin{definition}
Let $A^*[i,j]$ denote the difference in area of an optimum $j$-compression $C$ of $R_{\overleftarrow{x_i}}$. Let $T^*[i,j]=i'$ denote the index where the last segment of $C$ starts (i.e., the last segment of $C$ starts at $x_{i'}\in X$ with $i'<i$; Lemma \ref{lem:technical}). Let $H^*[i,j]=p$ denote the pitch of last segment of $C$.
\end{definition}

Thinking in dynamic programming terms, for every $1\leq i\leq n$, consider $A_i$ and $T_i$ as $k\times \rho$ tables. Similarly, $A^*$, $T^*$ and $H^*$ are $n\times k$ tables. Then, from Lemma~\ref{lem:dp-1}, the following result is deduced.

\begin{corollary}\label{cor:dp}
For a fixed value $i$, $1\leq i \leq n$, assume that tables $A_i$ and $T_i$ are already computed. Then the $i$-th row of tables $A^*$, $T^*$ and $H^*$ can be computed in $O(k\rho)$ time (where $\rho=|P|$) as follows.

Let $p^*\in P$ be the minimum pitch value such that $A_i[j,p^*]=\min_{ p\in P}A_i[j,p]$. Then
\[H^*[i,j]=p^*,\;\; A^*[i,j]=A_i[j,p^*]\;\text{ and }\;T^*[i,j]=T_i[j,p^*].\]

\end{corollary}

Finally, the following lemma establishes the dynamic programming recurrence.

\begin{lemma}\label{lem:dp-2}
Fix a value $1\leq i\leq n$. Assume that tables $A_i$ and $T_i$ are already computed. Assume that the $i$-th row of tables $A^*$ and $T^*$ are already computed as well. Then tables $A_{i+1}$ and $T_{i+1}$ can be obtained as follows:
\begin{gather*}
A_{i+1}[j,p]=\min\{A_i[j,p]+A^{(p)},A^*[i,j-1]\} \\
T_{i+1}[j,p]=\left\{
\begin{array}{cl}
    T_{i}[j,p] & \text{if\; $A_i[j,p]+A^{(p)}\leq A^*[i,j-1]$,} \\
    i & \text{otherwise,}
\end{array}
\right.
\end{gather*}
where $A^{(p)}$ denotes the area between the segment $[x_i,x_{i+1}]$
of $R$  and the last
compressing segment with pitch $p$,
(i.e., $A^{(p)}=\left|p-R(x_i)\right|\cdot(x_{i+1}-x_i)$)
\end{lemma}

\begin{proof}
Let $C$ be the melody with minimum difference in area from $R_{\overleftarrow{x_{i+1}}}$ formed by $j$ segments starting at time $x_0$ and ending at time $x_{i+1}$ whose last segment has pitch $p\in P$ and starts at some time $x\in X$ with $x<x_{i+1}$.

If the last segment of $C$ starts at time $x_{i'}<x_i$, then the difference in area between $C_{\overleftarrow{x_i}}$ and $R_{\overleftarrow{x_i}}$ is $A_i[j,p]$. Therefore, the difference in area between $C$ and $R_{\overleftarrow{x_{i+1}}}$ is $A_{i+1}[j,p]=A_i[j,p]+A^{(p)}$ and its last segment starts at time $T_{i+1}[j,p]=T_i[j,p]=i'$. In this case, the last segment of the previous compression is extended to  $x_{i+1}$.

If the last segment of $C$ starts at time $x_i$, then obviously $T_{i+1}[j,p]=i$. Moreover, by Lemma~\ref{lem:compression-prefix} we have that $C_{\overleftarrow{x_i}}$ is a $(j-1)$-compression of $R_{\overleftarrow{x_i}}$ with minimum difference in area $A^*[i,j-1]$. Therefore, the difference in area between $C$ and $R_{\overleftarrow{x_{i+1}}}$  is $A_{i+1}[j,p]=A^*[i,j-1]$ and the result follows.
\end{proof}

Using Corollary \ref{cor:dp} and Lemma \ref{lem:dp-2} we are ready to implement a dynamic programming algorithm to fill tables $A^*$, $T^*$ and $H^*$, from which we can obtain the $k$-compression with minimum difference in area.

\begin{theorem}
    An optimum $k$-compressed melody of $R$ can be found in $O(k\rho n)$ time and using $O(kn)$ memory space, where $n$ is the number of segments in  melody $R$ and $\rho$ is the number of different pitch values of the segments in $R$.
\end{theorem}
\begin{proof}
Note that the size of tables $T_i$ and $A_i$ is $O(k\rho)$, and computing a cell of these tables takes constant time assuming that the values of tables $T_{i-1}$ and $A_{i-1}$ as well as the values of the cells in the $(i-1)$-th row of $T^*$, $A^*$ and $H^*$ are known.

Now, the size of tables $T^*$, $A^*$ and $H^*$ is $O(kn)$ and computing a cell of the $i$-th row of these tables takes $O(\rho)$ time assuming that we already know the values of tables $T_i$ and $A_i$.
Since we need to repeat the process for $i=1,\dots,n$, the total computation time is $O(k \rho n)$.

Note that the difference in area from the optimal $k$-compression is in $A^*[k,n]$ and it can be obtained in $O(k)$ time by navigating backwards using the tables $T^*$ and $H^*$ (i.e., the last segment starts at $x_{i_k}$ time where $T^*[k, n]=i_k$ and its pitch is $H^*[k, n]$, the $(k-1)$-th segment starts at $x_{i_{k-1}}$ time where $T^*[k-1, i_k]=i_{k-1}$ and its pitch is $H^*[k, i_k]$, etc.).

Finally, tables $A_1,A_2,\dots,A_n$ can be modeled as a single table of size $k\times \rho$ that updates its cell values for $i=1,2\dots,n$. The same occurs for $T_1,T_2,\dots,T_n$. Then these two sequences of tables can be modeled using two single tables of size $k\times \rho$, which require $O(k\rho )$ memory space. Tables $T^*$, $A^*$ and $H^*$ requires $O(kn)$ memory space. Taking into account that $\rho \leq n$, the result follows.
\end{proof}

\section{Conclusions}\label{sec:conclu}

In this work, we developed efficient algorithms for
two geometric problems that arise in MIR; linear scaling and data compression.
The geometric nature of the problem comes from the use of two geometric metrics to measure the similarity between two melodies.
Our algorithms take advantage of the properties of an optimal solution and are based on line sweeping and dynamic programming techniques.
For a segment representation of the melodies, we considered the difference in area between the melodic contours and for the point representation we introduced a constrained matching, the $t$-monotone matching, in which a point of a melody is assigned to its neighbor to the left or to the right of the other melody.
Both difference in area and geometric matching have been considered in the literature.

Our future work will concentrate on experiments with real music data to explore the behavior of these similarity measures. In the experiments, the definition in the scaling operation can be extended by considering a different way of increasing the lengths of the segments. For example, in an $\varepsilon$-scaling, instead of increasing $\varepsilon$ the lengths of all segments,  we could increase each segment by $\alpha \varepsilon$, where the parameter $\alpha\in [0,1]$ depends on the length of the segment.
For the compression problem, we plan to explore the accuracy, flexibility and speed of
of our techniques for clustering and classification tasks.


%
%
%

\vspace{2mm}
\noindent
{\bf Acknowledgment.}
This work was initiated at the VIII
Spanish Workshop on Geometric Optimization, El Roc\'io,
Huelva, Spain, June 20--24, 2016. We thank all
participants in the workshop, especially Jorge Urrutia, for helpful discussions
and their contribution to a creative atmosphere.

\bibliography{referencias}

\end{document}